\documentclass[review,hidelinks,onefignum,onetabnum]{siamart251216}

\usepackage{lipsum}
\usepackage{amsfonts}
\usepackage{graphicx}
\usepackage{epstopdf}
\usepackage{tikz}
\usepackage{algorithmic}
\ifpdf
  \DeclareGraphicsExtensions{.eps,.pdf,.png,.jpg}
\else
  \DeclareGraphicsExtensions{.eps}
\fi

\newsiamremark{remark}{Remark}
\newsiamremark{exmp}{Example}
\newsiamremark{hypothesis}{Hypothesis}
\crefname{hypothesis}{Hypothesis}{Hypotheses}
\newsiamthm{claim}{Claim}
\newsiamremark{fact}{Fact}
\crefname{fact}{Fact}{Facts}

\headers{Computation of entanglement by a CBO method}{M. Herty, 
Y. Tang, and 
Y. Zhou}

\title{Computation of entanglement for quantum states by a Consensus-Based Optimization method\thanks{ 
\funding{This work was funded by the Deutsche Forschungsgemeinschaft (DFG, German Research Foundation) through 
HE5386/33-1 Control of Interacting Particle Systems, and Their Mean-Field, and Fluid-Dynamic Limits (560288187) and through 
HE5386/34-1 Partikelmethoden für unendlich dimensionale Optimierung (561130572) and through  the Priority Programme SPP2298 Theoretical Foundations of Deep Learning.}}}

\author{Michael Herty \thanks{Chair in Numerical Analysis, IGPM, RWTH Aachen University, Im Süsterfeld 2, D-52072 Aachen, Germany (herty@igpm.rwth-aachen.de); Extraordinary Professor, Department of Mathematics and Applied Mathematics, University of Pretoria, Private Bag X20, Hatfield 0028, South Africa}
\and Yijia Tang \thanks{Corresponding author, Scientific Computing Center, Karlsruhe Institute of Technology, 76131 Karlsruhe, Germany (yijia.tang@kit.edu)}
\and Yizhou Zhou \thanks{IGPM, RWTH Aachen University, Im Süsterfeld 2, D-52072 Aachen, Germany (zhou@igpm.rwth-aachen.de)}}

\usepackage{amsopn}

\ifpdf
\hypersetup{
  pdftitle={Computation of entanglement for quantum states by a Consensus-Based Optimization method},
  pdfauthor={Michael Herty, 
Yijia Tang, and 
Yizhou Zhou}
}
\fi

\begin{document}

\maketitle

\begin{abstract}
The computation of quantum entanglement can be formulated as a high-dimensional nonconvex optimization problem with orthogonality constraints. In this work, we propose structure-preserving consensus-based optimization (CBO) methods for entanglement computation, with one approach based on a Hermitian formulation and the other evolving directly on the unitary manifold. To handle the variable dimension of the feasible set, we introduce a cross-dimensional interaction mechanism allowing exchange of information between particles of different sizes. Numerical experiments demonstrate that the proposed methods achieve accurate approximations.
\end{abstract}



\section{Introduction}

Quantum information theory plays a central role in modern physics, quantum computing, and information science \cite{MR1796805}. A key challenge in this field is the characterization and computation of quantum entanglement \cite{BennettDiVincenzoSmolinWootters96,MR2515619,GUHNE20091,PhysRevA.83.062325}, which serves as a fundamental resource for quantum teleportation, cryptography and computation. However, many entanglement-related problems are still open and known to be computationally intractable in general \cite{PhysRevLett.95.210501,10.5555/2011350.2011361,MR2087945,TerhalVollbrecht00}. This intrinsic computational complexity makes it essential to develop efficient optimization methods capable of handling high-dimensional quantum systems.

To address these challenges, a variety of numerical and analytical approaches have been proposed in the literature. Classical methods include gradient-based optimization techniques \cite{PhysRevA.80.042301,RyuCaiCaro08}, variational approaches \cite{PhysRevA.64.052304}, Semidefinite Programming \cite{Peres96,Rains,DPS}, machine learning-based method \cite{958b-87zr}, and simulated annealing \cite{Zyczkowski99} have been employed to explore the highly nonconvex landscape associated with entanglement optimization. Despite these advances, accurately computing entanglement in large-scale quantum systems remains a difficult task, and existing methods often suffer from either high computational cost or limited accuracy. 

In this work, we focus on the computation of entanglement using a Consensus-Based Optimization (CBO) method \cite{MR3597012}. CBO is a recently developed stochastic particle-based optimization framework designed for nonconvex high-dimensional problems. The main idea is to evolve a system of interacting particles toward a consensus point that approximates a local minimum of the system. Due to its simplicity, parallel structure, and robustness, CBO has been successfully applied to a wide range of optimization problems, including machine learning \cite{MR4621974,CarrilloJinLiZhu21,6483403}, constrained optimization \cite{MR4337731,MR4538901}, stochastic optimization \cite{MR4990453}, multi-objective optimization \cite{MR4629470}, qubit configuration optimization \cite{de2025consensus} and so on. Moreover, a rigorous convergence analysis of CBO in the many-particle limit has been developed in recent years \cite{MR3804923,MR4793478,MR4179193,MR4215338,MR4404221,MR4947952,huang2026faithfulglobalconvergencerescaled,MR4456068}. Motivated by these advantages, we adapt the CBO framework to the problem of entanglement computation. A naive application of CBO to this problem is however infeasible due to the distinct features of entanglement computations.

Compared to other nonconvex optimization problems, the computation of entanglement exhibits a distinctive feature: the feasible set is given by complex finite dimensional matrices lying on a nonlinear complex Stiefel manifold,
$$
V_r(\mathbb{C}^M)=\{U\in \mathbb{C}^{M\times r}:U^\dagger U = I_r\}.
$$ 
Due to this geometric constraint, it is both important and natural to design a CBO method that preserves the orthogonality structure throughout the evolution. So far, classical CBO methods have mainly been applied in $\mathbb{R}^d$ and typically do not account for the underlying manifold. At this point, we comment on the works \cite{MR4404221,HaIEEE} about the CBO method on the real Stiefel manifold with rigorous mean-field analysis. In contrast to their predictor–corrector scheme and Cayley transformation scheme, our methods are based on two different frameworks: (1) We first present a Hermitian-preserving CBO system that directly evolves Hermitian matrices. Then we construct unitary matrices via the exponential map. (2) Secondly, we compare this to a novel method where the CBO system preserves unitary matrices. To this end, the drift and diffusion needs to be tailored to the given setup.
At the discrete level, the unitary property strictly preserves by resorting to an exponential integrator. 

We provide a brief discussion of the two methods.
Our Hermitian-preserving CBO method can be reformulated within the vector-based framework in \cite{CarrilloJinLiZhu21}, thereby inheriting its solid theoretical foundation, including convergence at the mean-field level.
The key distinction is that we explicitly design the matrix structure to preserve the Hermitian property. Notably, this approach is restricted to $M\times M$ square matrices, so a truncation is required to extract the relevant $M\times r$ component. In contrast, the unitary-preserving method evolves the $M\times r$ matrix directly without truncation. The numerical experiments in Section \ref{sec:5} demonstrate the performance of both methods in several benchmark problems.

Furthermore, we propose a new multi-species CBO method that enables information exchange among different species. This idea is novel compared to classical CBO methods \cite{CarrilloJinLiZhu21,MR4404221,MR3597012} and utilizes the intrinsic structure of entanglement computation. In our setting, different species correspond to different matrix dimensions $M$ in the feasible set of the optimization problem.
In practice, problems associated with different values of $M$ are typically solved independently. We instead introduce a mechanism that allows communication across different values of $M$. The key idea is to compute the consensus using information from particles with different dimensions. This cross-dimensional interaction is realized through a embedding and truncation procedure, which makes it possible to combine matrices of different dimensions $M$ in a consistent manner and exploit this additional information in e.g. computation of the consensus-point. As a result, the resulting coupled dynamics enhance the overall exploration capability of the method and provide a more flexible framework for entanglement optimization.

The paper is organized as follows. In Section~\ref{sec:2}, we briefly review the basic concepts of quantum entanglement and the standard Consensus-Based Optimization method. Section~\ref{sec:3} is devoted to two structure-preserving CBO approaches. In Section~\ref{sec:4}, we introduce a multi-species CBO method for information exchange among particles with different dimensions. Finally, numerical experiments are presented in Section~\ref{sec:5}.

\section{Preliminaries}\label{sec:2}
In this section, we first introduce the background of entanglement for quantum states and, in particular, present the optimization problem for computing the entanglement of formation. In the second subsection, we introduce the typical Consensus-Based Optimization method.

\subsection{Entanglement for quantum states}
In this work, we consider a bipartite mixed quantum state consists of two subsystems of dimension $N_A, N_B$. The density matrix $\rho_{AB}$ is a $N\times N$ positive semi-definite Hermitian matrix with trace one. Here $N=N_AN_B$. 

The density matrix of a mixed quantum state can be expressed in terms of pure-state decomposition 
\begin{equation}\label{rhoAB}
\rho_{AB} = \sum_{k=1}^r p_k |\psi_k\rangle \langle \psi_k|,\qquad \sum_{k=1}^r p_k = 1,~~p_k \ge 0,
\end{equation}
where $|\psi_k\rangle$ represents a unit column vector in $\mathbb{C}^N$, $\langle \psi_k|$ is the conjugate transpose of $|\psi_k\rangle$.
Physically, the quantum system is in the pure state $|\psi_k\rangle$ with probability $p_k$.
Moreover, the matrix expression for \eqref{rhoAB} reads
$$
\rho_{AB}= \Psi P \Psi^{\dagger}
$$
Here $\dagger$ means the conjugate transpose, $P=\text{diag}(p_1,\dots,p_r)$ and the column vectors of $\Psi$ are $|\psi_k\rangle$ with $1\leq k\leq r$.



Clearly, the spectral decomposition of $\rho_{AB}$ provides one possible pure-state decomposition. Namely, we take $P$ and $\Psi$ to be the eigenvalues and eigenvectors of $\rho_{AB}$, and let $r$ denote its rank. 

However, the pure-state decomposition is not unique. For $M\geq r$, let $U\in\mathbb{C}^{M\times r}$ be a matrix satisfying 
$U^{\dagger}U=I_r$. Then we have 
$$
\rho_{AB}= \Psi P^{1/2} U^{\dagger} U P^{1/2} \Psi^{\dagger} = WW^{\dagger}
$$
with 
\begin{equation}\label{W-U}
    W=\Psi P^{1/2} U^{\dagger}.
\end{equation} 
Denoting $|w_m\rangle$ the $m$-th column of $W$ for $1\leq m\leq M$, we have 
\begin{equation}\label{W-phi}
\rho_{AB} = \sum_{m=1}^M |w_m\rangle \langle w_m|  = \sum_{m=1}^M \sigma_m |\phi_m\rangle \langle \phi_m|
\end{equation}
with $\sigma_m=\langle w_m | w_m\rangle$ and $|\phi_m\rangle=| w_m\rangle / \sqrt{\sigma_m}$. The formula \eqref{W-phi} provides a further pure-state decomposition. 
Through this expression, the freedom of the decomposition can be characterized in terms of the matrix $U$.

For each $U$, we compute the entanglement $E_{AB}(U)$ by the following steps: \cite{BennettDiVincenzoSmolinWootters96,Wootters98}\\
\begin{itemize}
    \item[(E1)] Compute $W$ according to \eqref{W-U}.
    \item[(E2)] Compute $\phi_m$ and $\sigma_m$ for $1\leq m\leq M$ by \eqref{W-phi}.
    \item[(E3)] Denote $\rho^m=|\phi_m\rangle\langle \phi_m|$. And compute the reduced state by taking the partial trace $\rho^m_A=\text{Tr}_B \rho^m$.
    \item[(E4)] Compute the entanglement of pure state by the von Neumann entropy $$E(|\phi_m\rangle)=-\text{Tr} \left(\rho^m_A\log\rho^m_A\right).$$
    \item[(E5)] Compute $E_{AB}(U)$ by:
    $$
    E_{AB}(U) = \sum_{m=1}^M\sigma_m E(|\phi_m\rangle).
    $$
\end{itemize}
Since there are infinitely many $U$, the Entanglement of Formation (EoF) \cite{BennettDiVincenzoSmolinWootters96} of $\rho_{AB}$ is defined as the minimal amount of entanglement required to create $\rho_{AB}$, that is,
\begin{equation}
\label{eq: EoF}
\mathrm{EoF}=\min_{U^{\dagger}U=I_r}E_{AB}(U).
\end{equation}
Note that $E_{AB}(U)$ is nonnegative for any $U$ due to the nonnegativity of the von Neumann entropy, and the same holds for EoF.

The EoF is a entanglement measure and is closely related to the question of whether a mixed bipartite state $\rho_{AB}$ is \emph{separable} or \emph{entangled}.
The mathematical formulation is given as follows: 
\begin{itemize}
    \item \emph{Separable (not entangled)} if there exists an finite integer $K\in\mathbb{N}^+$ and $K$ constants $\gamma_k>0$ for $1\leq k\leq K$, such that
    \[
    \rho_{AB} = \sum_{k=1}^K \gamma_k \, \rho_A^{(k)} \otimes \rho_B^{(k)}, \qquad \sum_{k=1}^K \gamma_k = 1.
    \]
    Here $\rho_A^{(k)}\in \mathbb{C}^{N_A\times N_A}$ and $\rho_B^{(k)}\in \mathbb{C}^{N_B\times N_B}$ are valid density matrices on subsystem $A$ and $B$, respectively.
    \item \emph{Entangled} if it is not \emph{Separable}.\\
\end{itemize}

Note that, a quantum state $\rho_{AB}$ is separable, if and only if there exists an finite integer $M\in\mathbb{N}^+$ and $M$ constants $\sigma_m>0$ for $1\leq m\leq M$, such that
\begin{equation*} 
    \rho_{AB} = \sum_{m=1}^M \sigma_m \, |\phi_m\rangle \langle \phi_m|,\qquad |\phi_m\rangle = |\alpha_m\rangle \otimes |\beta_m\rangle.
\end{equation*}
Here $\sum_i \sigma_i = 1$, states $|\alpha_m\rangle\in \mathbb{C}^{N_A}$ and $|\beta_m\rangle\in \mathbb{C}^{N_B}$. Due to the property of $|\phi_m\rangle$, we have $\rho_A^m = |\alpha_m\rangle \langle \alpha_m|$ when computing the entanglement (see Step (E3)).
Thus we have $E(|\phi_m\rangle)= 0$ for every $m$. Clearly, it leads to
$E_{AB}(U)=0$. Therefore, to determine whether a quantum state $\rho_{AB}$ is separable, we need to solve the optimization problem \eqref{eq: EoF} and check whether the EoF attains zero.

\subsection{Consensus-Based Optimization (CBO)}
Consensus-Based Optimization (CBO) \cite{MR3597012} is a particle-based optimization method designed to approximate the global minimizer of an objective function, i.e.
$$
\min_{x\in\mathbb{R}^d}f(x).
$$
The CBO method does not rely on gradient information and it uses a system of interacting particles
\[
dX_t^j = -\lambda \bigl(X_t^j - x_\beta(t)\bigr)\,dt
+ \sigma \lvert X_t^j - x_\beta(t) \rvert \, dW_t^j,
\quad j = 1, \dots, J,
\]
where \(\lambda > 0\) is the drift rate, \(\sigma > 0\) is the noise intensity, and \((W_t^j)_{j=1}^J\) are independent $d$-dimensional standard Brownian motions. 

Each particle represents a candidate solution, and its movement is influenced by both deterministic attraction toward a collective consensus and stochastic fluctuations that encourage exploration. The particles system is connected to the optimization by the \emph{consensus point}, defined as the weighted average
\[
x_\beta(t) =
\frac{\sum_{j=1}^{J} X_t^j \, \exp(-\beta f(X_t^j))}
{\sum_{j=1}^{J} \exp(-\beta f(X_t^j))},
\]
where \(\beta \gg 1 \) is a parameter. In this way, particles located in regions where \(f\) is small contribute more significantly to the average. This Gibbs distribution type weight associated with the objective function is justified by the Laplace principle \cite{DemboZeitouni}.

The original CBO system uses isotropic diffusion, where all dimensions are explored equally. In \cite{CarrilloJinLiZhu21}, an anisotropic diffusion is introduced, namely, the component-wise geometric Brownian motion is used so that each dimension is explored with a different magnitude. This approach is better suited for high-dimensional optimization problems, as the constraint between $\lambda$ and $\sigma$ for forming a consensus among particles is insensitive to the dimension $d$.

Under suitable conditions, it can be shown that the particles tend to collapse toward a common point that approximates the global minimizer of $f$. This collective behavior can be analyzed in a mean-field limit as $J\rightarrow \infty$, where the particle system is described by a kinetic or Fokker–Planck type equation governing the evolution of the particle density \cite{MR3804923,MR4793478,MR4179193,MR4404221,MR4456068,MR4947952,huang2026faithfulglobalconvergencerescaled,MR4215338}. These theoretical properties, together with the simplicity of the algorithm and its gradient-free nature, make CBO an attractive approach for high-dimensional optimization problems.

\section{Structure-preserving CBO method for the entanglement}\label{sec:3}

In this section, we modify the CBO on the constrained  optimization problem \eqref{eq: EoF}. In our setting, the optimization variable is a complex matrix 
$U\in\mathbb{C}^{M\times r}$ subject to the orthogonal constraint. The feasible set is no longer a linear space but a nonlinear Stiefel manifold 
$$
V_r(\mathbb{C}^M)=\{U\in \mathbb{C}^{M\times r}:U^\dagger U = I_r\}.
$$ 
To incorporate this structure within the CBO framework, one should evolve the particle system by enforcing the constraint at each iteration. To this end, two different ways to ensure that the CBO dynamics remain in $V_r(\mathbb{C}^M)$ are possible.

\subsection{Hermitian-preserving CBO method}\label{sec3.1}
Given an $M \times M$ Hermitian matrix $H$, define
$$
\widetilde{U} = \exp(iH).
$$
It is straightforward to verify that $\widetilde{U}$ is unitary. Then, $U \in \mathbb{C}^{M \times r}$ is constructed by taking the first $r$ columns of $\widetilde{U}$. By construction, $U$ lies on the Stiefel manifold $V_r(\mathbb{C}^M)$.

Thanks to this observation, the original task reduces to ensuring the Hermitian property of $H$ in the implementation of the CBO algorithm. To this end, we evolve $J$ particles $H^j$, $1 \leq j \leq J$, where each $H^j$ is an $M\times M$ Hermitian matrix. The $j$-th particle $H^j$ satisfies the CBO equation:
\begin{equation}
\label{eq: CBO_H1}
\left\{
\begin{aligned}
	dH^j(t)=& -\lambda(H^j-\bar{H})dt+\sigma(H^j - \bar{H}) \odot dW,  \\[2mm]
    H^j(0)=&~H^{j,0}.
\end{aligned}
\right. 
\end{equation}
Here $\odot$ represents the Hadamard-product, so anisotropic diffusion is utilized. $\bar{H}$ is the consensus computed by
\begin{equation}\label{consensus-H}
\bar{H} = \frac{1}{\sum_{j=1}^J\exp(-\beta E_{AB}(U^j))}\sum_{j=1}^JH^j\exp(-\beta E_{AB}(U^j))
\end{equation}
with $\beta>0$ a sufficiently large constant and $U^j$ being the first $r$-columns of $\widetilde{U}^j = \exp(iH^j)$. The drift rate $\lambda$ and the noise parameter $\sigma$ are  constants. Moreover, $W$ is an $M \times M$ matrix-valued stochastic process representing the perturbation. 

It is easy to see that $\bar{H}$ is always Hermitian. 
The key to preserving the Hermitian property of $H^j$ during the evolution is to ensure that the noise $dW$ is also Hermitian. To this end, we take $\{W_{m_1,m_2}\}_{1\leq m_1\leq m_2 \leq M}$ as $M(M+1)/2$ independent complex random processes, whose real and imaginary parts are standard Brownian motion.
Here $W_{m_1,m_2}$ is the $(m_1,m_2)$-th entry of $W$.
For $m_1 > m_2$, we define $W_{m_1,m_2}$ as the complex conjugate of $W_{m_2,m_1}$, ensuring that $W$ is  Hermitian. 

Now we present the first algorithm to implement the Hermitian-preserving CBO, see Algorithm 3.1.

\begin{remark}
The equation~\eqref{eq: CBO_H1} is written in matrix form. Through a one-to-one mapping, it can equivalently be expressed in vector form. Since $H^j$ is a complex Hermitian matrix, the unknowns in~\eqref{eq: CBO_H1} consist of the real parts $\mathrm{Re}(H^j_{m_1,m_2})$ for $m_2 \geq m_1$ and the imaginary parts $\mathrm{Im}(H^j_{m_1,m_2})$ for $m_2 > m_1$. 
By our construction,~\eqref{eq: CBO_H1} is therefore equivalent to evolving these $M^2$ independent variables, from which $H^j$ can be reconstructed using the Hermitian property. This vector formulation enables a direct comparison with the formulation in \cite{CarrilloJinLiZhu21}, which uses component-wise Brownian motion.  
Particularly, the theoretical results in \cite{CarrilloJinLiZhu21} guarantee the convergence of the CBO method in the  mean-field limit level with parameter constraints independent of the dimensionality.
\end{remark}



\begin{remark}
We compute the entanglement for each value of $M$. According to the theoretical result \cite{Horodecki97}, the lower and upper bounds are $r$, the rank of $\rho_{AB}$, and $N^2$, the square of the matrix dimension, respectively.
\end{remark}

\begin{remark}\label{remark-projH}
Another natural method to guarantee the Hermitian property under the framework of CBO is to exploit projection at each step. Namely, we consider the scheme \begin{equation}
\label{eq: CBO-H-p}
\left\{
\begin{aligned}
H^{j,k+1/2} =&~ H^{j,k} - \lambda\Delta t (H^{j,k}-\bar{H}^k) + \sigma\sqrt{\Delta t}(H^{j,k}-\bar{H}^k)\odot Z^k
,\\[2mm]
H^{j,k+1}=&~\frac{H^{j,k+1/2}+({H^{j,k+1/2}})^\dagger}{2}.
\end{aligned}
\right.
\end{equation}
Here the $(m_1,m_2)$-th entry of $Z^k$ is given by $Z^k_{m_1,m_2}\sim \mathcal{N}(0,1)$.
The second step is used to impose the Hermitian constraint.
\end{remark}

\subsection{Unitary-preserving CBO method}

Previously, we solved an optimization problem for a full $M\times M$ matrix, although only the first $r$ columns (with $r\leq M$) are required. Moreover, the approach involved first computing a Hermitian matrix and then constructing a unitary matrix. In contrast, the new method directly computes the desired $M\times r$ matrix whose columns are orthogonal.

The key issue is to ensure that the matrix $U\in \mathbb{C}^{M\times r}$ lies in the Stiefel manifold, i.e., it satisfies the orthonormality constraint $U^\dagger U=I_r$. This structural property should be preserved throughout the computation. To this end, we consider the stochastic differential equation (SDE)
\begin{equation}\label{CBO-unitary}
\left\{
\begin{aligned}
	dU^j(t)=& ~ \lambda(\bar{U}(U^j)^\dagger - U^j \bar{U}^\dagger)U^j ~dt+\sigma dW \circ U^j, \\[2mm]
    U^j(0)=&~U^{j,0}.
\end{aligned}
\right.
\end{equation}
Notice that the SDE should be interpreted in the Stratonovich sense.
Here $dW$ is designed to be a skew-Hermitian matrix (details are given in \eqref{skew-symm-dW}). The consensus $\bar{U}$ is given by
$$
\bar{U}=\text{argmin}_jE_{AB}(U^j).
$$
Clearly, $\bar{U}$ also satisfies 
$\bar{U}^\dagger\bar{U}=I_r$.
The initial data $U^{j,0}\in \mathbb{C}^{M\times r}$ consists of $r$ orthogonal $\mathbb{C}^m$ vectors, $\lambda$ and $\sigma$ are scalars. 

\begin{remark}
To ensure that $\bar{U}$ remains on the Stiefel manifold, the standard CBO averaging
$$
\bar{U} = \frac{\sum_{j=1}^J \omega_j U^j}{\sum_{j=1}^J \omega_j}, 
\qquad \omega_j = \exp(-\beta E_{AB}(U^{j})),
$$
cannot be used, as it does not preserve the orthogonality constraint.

Alternatively, for the unitary matrix $U^j\in\mathbb{C}^{M\times M}$, we could instead consider the weighted geometric mean
$$
\bar{U} = \prod_{j=1}^J (U^j)^{\eta_j}
= \prod_{j=1}^J \exp\bigl(\eta_j \log U^j\bigr),
\qquad 
\eta_j = \frac{\exp(-\beta E_{AB}(U^{j}))}{\sum_{j=1}^J \exp(-\beta E_{AB}(U^{j}))},
$$
which preserves the unitary property. 
\end{remark}

We first show that the dynamic \eqref{CBO-unitary} preserves the unitary property.
\begin{proposition}
Assume that the initial data $U^{j,0}$ satisfies $(U^{j,0})^\dagger U^{j,0}=I_r$. Then for any $t>0$, the matrix $U^j=U^j(t)$ satisfies 
$$
(U^j)^\dagger U^j = I_r.
$$
\end{proposition}

\begin{proof}
We denote 
    $$
    Q = \lambda(\bar{U}(U^j)^\dagger - U^j \bar{U}^\dagger)dt+\sigma dW.
    $$ Notice that $Q$ is a skew-Hermitian matrix since $\bar{U}(U^j)^\dagger - U^j \bar{U}^\dagger$ and $dW$ are both skew-Hermitian. Namely, $Q^\dagger=-Q$. Then in  the Stratonovich sense, we have
    \begin{align*}
        d((U^j)^\dagger U^j) = &~(d(U^j)^\dagger) U^j+(U^j)^\dagger d(U^j) 
        = ((U^j)^\dagger Q^\dagger )U^j+(U^j)^\dagger (QU^j) \\[1mm]
        =&~ -(U^j)^\dagger QU^j+(U^j)^\dagger QU^j = 0.
    \end{align*}
It follows from $(U^j)^\dagger(0)U^j(0)=I_r$ that
$$
(U^j)^\dagger(t)U^j(t)=I_r
$$
for any $t\geq 0$.
\end{proof}

\begin{algorithm}\label{algoH}
\caption{Hermitian-preserving CBO method}
\textbf{Initialize:} Given the density matrix $\rho_{AB}$\\[2mm]
Set the parameters $\beta$, $\lambda$ and $\sigma$. Compute the rank of $\rho_{AB}$ and denote as $r$. \\
Set the dimension $M\geq r$, the number of particles $J$, the iteration times $K$ and the time step $\Delta t$. \\[1mm]
Prescribe the initial value by generating Hermitian matrix $H^{j,0}~(1\leq j\leq J)$ randomly.\\[2mm]
\textbf{Updating in time:} for $0\leq k\leq K$: 
$$
k \Delta t \rightarrow (k+1)\Delta t,\qquad H^{j,k}\rightarrow H^{j,k+1} $$
\textbf{Step 1:} Compute the entanglement $E_{AB}(U^{j,k})$ by (E1)--(E5) with $U^{j,k}$ the first $r$ columns of $\exp(iH^{j,k})$. Then compute the consensus by
$$
\bar{H}^k = \frac{1}{\sum_{j=1}^J\exp(-\beta E_{AB}(U^{j,k}))}\sum_{j=1}^JH^{j,k}\exp(-\beta E_{AB}(U^{j,k}))
$$

\textbf{Step 2:} Update $H^{j,k+1}$ by the Euler-Maruyama scheme:
$$
H^{j,k+1} = H^{j,k} - \lambda\Delta t (H^{j,k}-\bar{H}^k) + \sigma\sqrt{\Delta t}(H^{j,k}-\bar{H}^k)\odot Z^k
$$
where the $(m_1,m_2)$-th entry of $Z^k$ is given by
$$
Z^k_{m_1,m_2} = \left\{
\begin{array}{ll}
   Y_{m_1,m_2}, &\quad m_1=m_2 \\[3mm]
   Y_{m_1,m_2}+i Y_{m_2,m_1}, &\quad m_1<m_2\\[3mm]
   Y_{m_1,m_2}-i Y_{m_2,m_1}, &\quad m_1>m_2
\end{array}
\right.
$$
with $Y_{m_1,m_2}\sim \mathcal{N}(0,1)$.

\textbf{Output:} Take $\bar{U}^k$ as the first $r$ columns of $\exp(i\bar{H}^k)$. The numerical results are 
$$
\min_{1\leq k\leq K} E_{AB}(
\bar{U}^k).
$$
\end{algorithm}

At the discrete level, we consider the exponential integrator \cite{MR3891437} to preserve the continuous property:
$$
U^{j,k+1}= \exp\left(\lambda \Delta t ~\left(\bar{U}^k(U^{j,k})^\dagger - U^{j,k} (\bar{U}^k)^\dagger\right)  +\sigma \sqrt{\Delta t}~Z^k \right)  U^{j,k}
$$
with $Z^k$ being a skew-symmetric matrix, whose $(m_1,m_2)$-th entry is given by
\begin{equation}\label{skew-symm-dW}
Z^k_{m_1,m_2} = \left\{
\begin{array}{ll}
   iY_{m_1,m_2}, &\quad m_1=m_2 \\[3mm]
   Y_{m_1,m_2}+i Y_{m_2,m_1}, &\quad m_1<m_2\\[3mm]
   -Y_{m_1,m_2}+i Y_{m_2,m_1}, &\quad m_1>m_2.
\end{array}
\right.
\end{equation}
with $Y_{m_1,m_2}\sim \mathcal{N}(0,1)$. 
Notice that the matrix 
$$
\exp\left(\lambda \Delta t ~\left(\bar{U}^k(U^{j,k})^\dagger - U^{j,k} (\bar{U}^k)^\dagger\right)  +\sigma \sqrt{\Delta t}~Z^k \right)  
$$
is unitary. Then we know that $U^{j,k+1}$ is unitary provided that $U^{j,k}$ is unitary.

According to the spirit of CBO method, the dynamics are composed of two key components: a drift term and a stochastic term. The drift term captures the deterministic part of the evolution and is responsible for guiding the system toward agreement among particles. In the following, we establish a proposition showing that the drift term drives the dynamics toward consensus intuitively.
\begin{proposition}
Consider the deterministic dynamic
$$
\frac{dU}{dt}=\lambda(\bar{U}U^\dagger - U \bar{U}^\dagger)U,\qquad U(0)=U_0 
$$
where $\bar{U}\in\mathbb{C}^{M\times r}$ is a given constant matrix with $M\geq r$. Then there exists a constant $\delta$ such that for all $U_0$ satisfying 
$\|U_0-\bar{U}\|<\delta$, the solution $U=U(t)\in\mathbb{C}^{M\times r}$ satisfies
$$
\lim_{t\rightarrow \infty} U(t)=\bar{U}.
$$
\end{proposition}

\begin{proof}
Denote $R=\bar{U}^\dagger U\in\mathbb{C}^{r\times r}$. Using the orthogonal property of $\bar{U}$ and $U$, we derive that 
$$
\frac{dR}{dt} = \lambda(I_n-R^2). 
$$
For the case $\|U_0-\bar{U}\|<\delta$, we know that there exists a small constant $C_\delta$ such that 
$\|R(0)-I_r\|=\|\bar{U}^\dagger U_0-I_r\|<C_\delta$. Now we write the linearized equation near the equilibrium $R_{\infty}=I_r$:
$$
\frac{d\tilde{R}}{dt} = -2\lambda\tilde{R}+O(\tilde{R}^2)
$$
with $R=\tilde{R}+R_{\infty}$. Since $-2\lambda<0$, by the theory of stable manifold \cite{MR69338}, we know that 
$$
\lim_{t\rightarrow \infty} \tilde{R}(t)= 0 .
$$
Therefore, we know that 
$$
\bar{U}^\dagger \lim_{t\rightarrow \infty}U(t) = I_r,
$$
which implies $\lim_{t\rightarrow \infty}U(t) = \bar{U}$.
\end{proof}

On the other hand, the stochastic term adds randomness to the dynamics, in order to explore different states and preventing it from getting stuck in a non-optimal solution. 
To summarize the above, we propose
the Unitary-preserving CBO method (see Algorithm \ref{algo2}).

\begin{algorithm}\label{algo2}
\caption{Unitary-preserving CBO method}
\textbf{Initialize:} Given the density matrix $\rho_{AB}$\\[2mm]
Set the parameters $\lambda$ and $\sigma$. Compute the rank of $\rho_{AB}$ and denote as $r$. \\
Set the dimension $M\geq r$, the number of particles $J$, the iteration times $K$ and the time step $\Delta t$. \\[1mm]
Prescribe the initial value by generating Hermitian matrix $H^{j,0}~(1\leq j\leq J)$ randomly. Then take $U^{j,0}$ as the first $r$-columns of $\exp(iH^{j,0})$. \\[2mm]
\textbf{Updating in time:} for $0\leq k\leq K$: 
$$
k \Delta t \rightarrow (k+1)\Delta t,\qquad U^{j,k}\rightarrow U^{j,k+1} 
$$
\textbf{Step 1:} Compute the entanglement $E_{AB}(U^{j,k})$ by (E1)--(E5). Then compute the consensus by
$$
\bar{U}^k = \text{argmin}_jE_{AB}(U^{j,k}).
$$

\textbf{Step 2:} Update $U^{j,k+1}$ by the exponential integrator 
$$
U^{j,k+1}= \exp\left(\lambda \Delta t ~\left(\bar{U}^k(U^{j,k})^\dagger - U^{j,k} (\bar{U}^k)^\dagger\right)  +\sigma \sqrt{\Delta t}~Z^k \right)  U^{j,k}.
$$

\textbf{Output:} The numerical results are 
$$
\min_{1\leq k\leq K} E_{AB}(
\bar{U}^k).
$$
\end{algorithm}

\begin{remark}\label{remark-projU}
In order to guarantee the orthogonal property, one can also use the projection method directly to the standard CBO system. 
By using the Euler-Maruyama scheme, the idea is implemented as
\begin{equation*}
\left\{
\begin{aligned}
U^{j,k+1/2} & = U^{j,k} - \lambda\Delta t (U^{j,k}-\bar{U}^k) + \sigma\sqrt{\Delta t}(U^{j,k}-\bar{U}^k)\odot Z^k,\\[2mm]
U^{j,k+1}~&\leftarrow~~\text{Gram-Schmidt process of}~ U^{j,k+1/2}.
\end{aligned}
\right.    
\end{equation*}
In this case,  the $(m_1,m_2)$-th entry of $Z^k$ is given by $Z^k_{m_1,m_2}\sim \mathcal{N}(0,1)$ for $1\leq m_1\leq M$ and $1\leq m_2\leq r$. Besides, $\bar{U}^k$ can also be constructed as the weighted sum:
$$
\bar{U}^k = \frac{\sum_{j=1}^J \omega_j U^{j,k}}{\sum_{j=1}^J \omega_j}, 
\qquad \omega_j = \exp(-\beta E_{AB}(U^{j,k})).
$$
Then we exploit the Gram-Schmidt process for $\bar{U}^{k}$. 

However, numerical experiments suggest that the projection method underperforms compared to the structure-preserving method (see Section \ref{sec:5}). A possible drawback is that the trajectory of the original CBO system may deviate substantially from the manifold, resulting in large projection steps, which can in turn distort the dynamics and lead to stability and consistency issues.
\end{remark}

\section{Multi-species CBO method}\label{sec:4}

For the computation of the EoF, the optimization problem should be solved for matrices of each order $M$, with $r \leq M \leq N^2$.
In the previous section, we developed Algorithms~3.1 and~\ref{algo2} to solve the problems independently for various $M$.

In this section, we explore the possibility to allow information exchange across different values of $M$, thereby enabling interaction between these problems. As a result, each system corresponding to a given $M$ can benefit from additional information, enhancing the ability of the CBO method to find the optimal solution. For simplicity, in the present work we use this coupling to the Hermitian-preserving framework introduced in Section~\ref{sec3.1}. 


Recall the consensus computation for the $M_1 \times M_1$ CBO system in Algorithm~3.1: 
$$
\bar{H}_{M_1} = \frac{1}{\sum_{j=1}^J\exp(-\beta E_{AB}(U^j_{M_1}))}\sum_{j=1}^JH^j_{M_1}\exp(-\beta E_{AB}(U^j_{M_1})).
$$
Here $\beta > 0$, and $U^j_{M_1}$ denotes the first $r$ columns of $\exp(i H^j_{M_1})$. The subscript $M_1$ emphasizes the dependence on the dimension. Similarly, one can compute the consensus for the $M_2 \times M_2$ CBO system. In the current formulation, there is no information exchange between the $M_1 \times M_1$ and $M_2 \times M_2$ systems. However, following the spirit of the CBO, if the objective value $E_{AB}(U^j_{M_2})$ is relatively small, then the corresponding $H^j_{M_2}$ should provide useful guidance not only for other particles of dimension $M_2 \times M_2$, but also for particles across all dimensions. We illustrate this idea in Fig \ref{fig:1}.

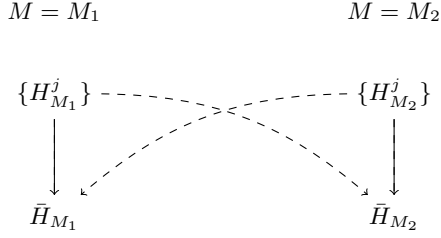
\begin{figure}[H]
    \centering
\begin{tikzpicture}[
    node distance=1.1cm and 1.0cm,
    every node/.style={font=\small}
]

\usetikzlibrary{positioning}

\node (M1) {$M=M_1$};
\node[right=3cm of M1] (M2) {$M=M_2$};

\node[below=0.5cm of M1] (P1) {$\{H^j_{M_1}\}$};
\node[below=0.5cm of M2] (P2) {$\{H^j_{M_2}\}$};

\node[below=of P1] (S1) {$\bar{H}_{M_1}$};
\node[below=of P2] (S2) {$\bar{H}_{M_2}$};

\draw (P1) -- (S1);
\draw (P2) -- (S2);

\draw[->, dashed] (P1) -- (S1);
\draw[->, dashed] (P2) -- (S2);

\draw[->, dashed] (P1) to[bend left=20] (S2);
\draw[->, dashed] (P2) to[bend right=20] (S1);

\end{tikzpicture}

    \caption{Consensus computation from multi-level particles}
    \label{fig:1}
\end{figure}

Specifically, to implement this idea, the first task is to address the summation of matrices $H^{j_1}_{M_1}$ and $H^{j_2}_{M_2}$ when they have different dimensions. Clearly, this requires truncating the larger matrix and augmenting the smaller one accordingly.
To achieve this, we state a simple fact about the problem of the entanglement computing.  
\begin{proposition}
Let $M_1 \in \mathbb{N}$ and let $U_1 \in \mathbb{C}^{M_1 \times r}$ satisfy $U_1^\dagger U_1 = I_r$. Denote its associated entanglement by $E_{AB}(U_1)$. Then, for any $M_2 \geq M_1$, there exists a matrix $U_2 \in \mathbb{C}^{M_2 \times r}$ such that $U_2^\dagger U_2 = I_r$ and
$$
E_{AB}(U_2) = E_{AB}(U_1).
$$
\end{proposition}
\begin{proof}
    According to the computation in (E1)--(E5), it suffices to give 
    $$
    U_2 = \begin{pmatrix}
        U_1\\
        0
    \end{pmatrix}.
    $$
    Indeed, by the formula \eqref{W-U}, we denote $W_1=\Psi P^{1/2} U_1^{\dagger} \in\mathbb{C}^{r\times M_1}$ and write
    $$
    W_1 = (w_1,w_2,\dots,w_{M_1}).
    $$
    It follows from \eqref{W-U} that $W_2\in \mathbb{C}^{r\times M_2}$ and 
    $$
    W_2=\Psi P^{1/2} U_2^{\dagger} = \Psi P^{1/2} (U_1^{\dagger},~0) = (w_1,w_2,\dots,w_{M_1},0,\dots,0).
    $$
    Then it is not difficult to see that $E_{AB}(U_1)=E_{AB}(U_2)$. 
\end{proof}

Motivated by this fact, for $M_1<M_2$ we define 
$$
\mathcal{T}_{M_1\rightarrow M_2}(H^{j}_{M_1}) := H^{j}_{M_1\rightarrow M_2} = \begin{pmatrix}
{H^{j}_{M_1}} & 0_{M_1 \times (M_2-M_1)} \\[2mm]
0_{(M_2-M_1) \times M_1} & 0_{(M_2-M_1) \times (M_2-M_1)}
\end{pmatrix}.
$$
Then we compute 
$$
\exp(iH^{j}_{M_1\rightarrow M_2}) = \begin{pmatrix}
\exp(iH^{j}_{M_1}) & 0 \\[2mm]
0 & I_{M_2-M_1}
\end{pmatrix}
$$
and its first $r$ columns are given by
$$
U^{j}_{M_1\rightarrow M_2} = \begin{pmatrix}
    U^{j}_{M_1}\\[1mm]
    0
\end{pmatrix}.
$$
We know that $E_{AB}(U^{j}_{M_1\rightarrow M_2})=E_{AB}(U^{j}_{M_1})$. In other words, the smaller matrix $H^{j}_{M_1}$ is embedded into the higher-dimensional space by zero padding.

On the other hand, there exist different methods to truncate a larger matrix to a smaller one. However, it is generally nontrivial to preserve the same entanglement after truncation.
In this work, we simply define $\mathcal{T}_{M_2\rightarrow M_1}(H^{j}_{M_2})$ by the following steps:
\begin{itemize}
\item[(T1)] Write \(H^{j}_{M_2} = (h_1, h_2, \dots, h_{M_2})\) in terms of its column vectors. Compute the Euclidean norm of each column \(\|h_m\|\), and select the indices \(m_1, m_2, \dots, m_{M_1}\) corresponding to the \(M_1\) largest norms.\\

\item[(T2)] Construct $\mathcal{T}_{M_2 \rightarrow M_1}(H^{j}_{M_2})$ by extracting the corresponding submatrix, whose $(p,q)$-th entry is given by $(H^{j}_{M_2})_{m_p, m_q}$.
\end{itemize}
~\\[2mm]
In this way, it is not difficult to check that 
$$
\mathcal{T}_{M_2\rightarrow M_1}(\mathcal{T}_{M_1\rightarrow M_2}(H^{j}_{M_1})) = H^{j}_{M_1}.
$$
Having the operators $\mathcal{T}_{M_2\rightarrow M_1}$ and $\mathcal{T}_{M_1\rightarrow M_2}$, we define the consensus for the $M\times M$ dimensional problem by
$$
\bar{H}_{M} = \frac{\sum_{m=r}^{N^2}\sum_{j=1}^J\omega_{jm}\mathcal{T}_{m\rightarrow M}(H^j_{m})}{\sum_{m=r}^{N^2}\sum_{j=1}^J \omega_{jm}} ,\qquad \omega_{jm}=\exp(-\beta E_{AB}(U^j_{m})).
$$
Here $U^j_{m}$ is an $M\times r$ matrix consists of the first $r$ columns of 
$\exp(i\mathcal{T}_{m\rightarrow M}(H^j_{m}))$.

~

We end this section with the following concluding remarks. 
\begin{remark}
We emphasize that the only modification with respect to the CBO method in Algorithm 3.1 lies in the definition of the consensus; all other steps remain unchanged. This method can be also  regarded as a CBO system on the extended phase space $(H,M)$ with no dynamics on $M$. 
\end{remark}
\begin{remark}
The proposed method is more general in nature. In this work, we only present one specific implementation, while there remain several freedoms. For example, one may allow particles to evolve across varying dimensions as the algorithm progresses; that is, $M$ is treated as a dynamical variable in the extended phase space. This would require additional rules for adapting the dimension according to the entanglement structure. Such an extension could further enhance the flexibility of the method and potentially improve its ability to explore the solution space.    
\end{remark}

\section{Numerical result}\label{sec:5}
In this section, we explore several benchmark problems in the computation of quantum entanglement to show the validity of our algorithms. 

\subsection{Numerical experiments for structure-preserving CBO methods}
In all the following examples, the CBO parameters are chosen as follows: inverse temperature $\beta=200$, drift rate $\lambda=1$, time step $\Delta t=0.2$, and maximum iteration number $K=1000$. Since the problem is highly non-convex, in order to better explore the landscape of the objective function, we add some addictive noise $\sigma \delta dW$ with $\delta = 1$. This comes in especially when the particles are about to reach a not optimal consensus. The convergence of this modified model has been investigated in \cite{MR4947952,huang2026faithfulglobalconvergencerescaled}.
Moreover, the initial data are generated by the Gaussian unitary ensemble (GUE) method \cite{Zyczkowski99,MR1628620}.

Though theoretical results \cite{Horodecki97} indicate that the parameter $M$ should lie in the range $[r,N^2]$. However, it has been shown in \cite{Zyczkowski99,RyuCaiCaro08} that only limited improvement in the optimal value is observed for large $M$, near its maximal value $N^2$. Thus, in practice we only compute the problem for $M\in [r, 2N]$.

\begin{exmp}[The Horodecki states $(2\times 2)$ \cite{Horodecki96}]
We consider the parameterized $2\times 2$ Horodecki state, whose density matrix $\rho_{AB}$ is given by
\begin{equation}
\rho_{AB}=q|\Psi_1 \rangle \langle \Psi_1|+(1-q)|\Psi_2 \rangle \langle \Psi_2|,
	\label{eq: Horodecki}
\end{equation}
where the parameters $0<q<1$, the pure states are given by
$$|\Psi_1 \rangle =a|00\rangle+\sqrt{1-a^2}|11\rangle,\quad |\Psi_2 \rangle =a|10\rangle+\sqrt{1-a^2}|01\rangle
$$
with $0<a<1$. In the following, we take $a = 3/4$. In this case, the rank for the density matrix $\rho_{AB}$ equals to $2$. Meanwhile, we have $N=N_A\times N_B=2\times 2 = 4$. 
\end{exmp}
For these 2-qubits states, the entanglement of formation is analytically given by the Wootters' formula \cite{HillWootters97,Wootters98}, resorting to concurrence $C(\rho_{AB})$.
$$
\mathrm{EoF}(\rho_{AB})=\mathcal{E}(C(\rho_{AB})),
$$
where 
$$
\mathcal{E}(x)=H\left(\frac{1}{2}+\frac{1}{2}\sqrt{1-x^2}\right), \quad x\in[0,1].
$$
Here $H(x)$ is the binary entropy defined by
$$
H(x)=-x\log_2 x - (1-x) \log_2 (1-x).
$$
Particularly, for the $2\times 2$ state, the concurrence is given by \cite{Wootters98} 
$$
C(\rho_{AB})=\max\{0,\lambda_1-\lambda_2-\lambda_3-\lambda_4\},
$$
where $\lambda_1\geq \lambda_2 \geq \lambda_3\geq \lambda_4$ are the eigenvalues of the Hermitian matrix 
$$
R=\sqrt{\sqrt{\rho_{AB}}~\tilde{\rho}_{AB}~\sqrt{\rho_{AB}}},\qquad \tilde{\rho}=(\sigma_y\otimes \sigma_y)\rho^*_{AB}(\sigma_y\otimes \sigma_y)
$$ 
with $\sigma_y$ being the Pauli Y matrix and $\rho^*_{AB}$ the complex conjugate of ${\rho}_{AB}$.

\begin{figure}
    \centering
    \includegraphics[width=0.8\linewidth]{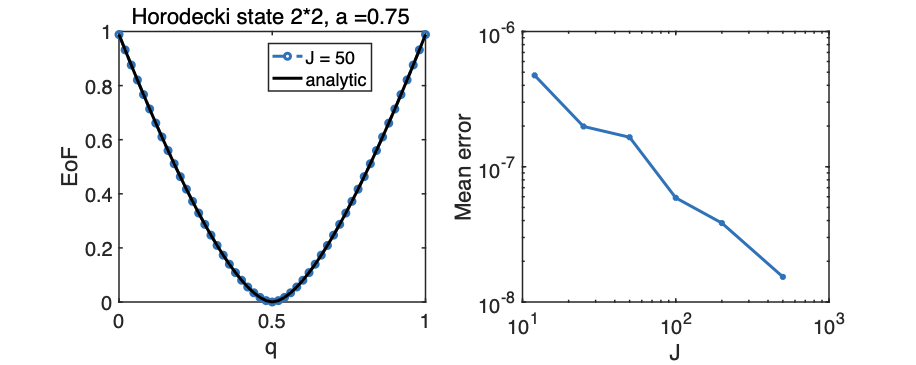}
    \caption{Left: EoF when $J = 50$, Right: Mean error for Horodecki states \eqref{eq: Horodecki} using Hermitian preserving CBO method w.r.t. particle number $J$.}
    \label{fig:conv_j}
\end{figure}

\begin{figure}
    \centering
    \includegraphics[width=0.4\linewidth]{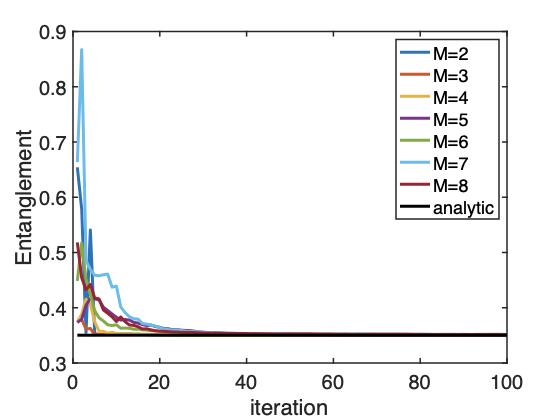}
    \caption{A trajectory of entanglement at consensus for different $M$ during iteration for $q=1/4$.}
    \label{fig:trajectory}
\end{figure}

 
To demonstrate the validity of the proposed CBO methods, we first compute the entanglement of formation (EoF) for different values of $q$ using the Hermitian-preserving CBO method. We verify that the computed $H^j$ preserves the Hermitian property up to machine precision.

The numerical results are shown in Fig. \ref{fig:conv_j}–\ref{fig:trajectory}. From Fig. \ref{fig:conv_j} (left), we observe that the numerical solutions agree well with the analytical solution. In Fig. \ref{fig:conv_j} (right), 
We compute the mean error of the EoF for each $J$, where $J$ denotes the number of particles. The error is averaged over all values of q. As the number of particles increases, the error decreases, indicating convergence of the $J$-particle system toward the mean-field limit. Moreover, we fix the parameter $q=1/4$ and plot the evolution of entanglement with respect to time in Fig. \ref{fig:trajectory}. We observe that the entanglement converges for all values of $M$, which implies that the CBO system reaches consensus as time $t$ becomes large.

Next, we compare the performances for structure-preserving CBO methods and the projection methods. 
The results are plotted in Fig. \ref{fig:horo3} with $J= 500$. In the left figure, we compare the performances for the Hermitian-preserving CBO method in Algorithm 3.1 and the projection method in Remark \ref{remark-projH}. In these two methods, we set the noise level to be the same $\sigma = 0.06$. On the other hand, the error for Unitary-preserving CBO and the projection method in Remark \ref{remark-projU} are compared in the right figure. Here the noise level is given as $\sigma = 0.01$. From these two figures, we observe that the structure preserving methods perform better than the projecting methods for both Hermitian and Unitary cases.

\begin{figure}
    \centering
\includegraphics[width=0.7\linewidth]{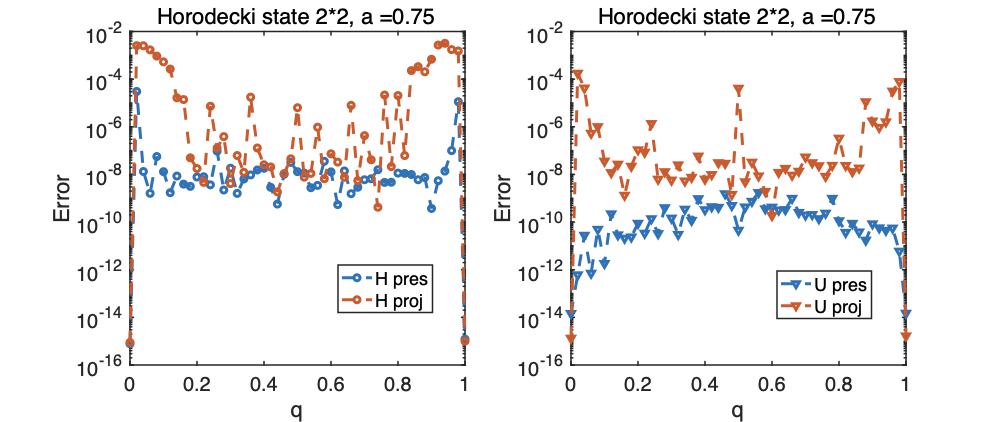}
    \caption{Comparison of structure preserving and projection CBO methods for Horodecki states in Hermitian (left) and Unitary (right) cases.}
    \label{fig:horo3}
\end{figure}

\begin{exmp}[The Werner states $(2\times 2)$ \cite{BennettDiVincenzoSmolinWootters96}]
Consider the parameterized $2\times 2$ Werner state. The density matrix $\rho_{AB}$ is given by 
\begin{equation}
	\rho_{AB}=F|\psi^-\rangle\langle\psi^-|+\frac{1-F}{3}\left(|\psi^+\rangle\langle\psi^+|+|\phi^-\rangle\langle\phi^-|+|\phi^+\rangle\langle\phi^+|\right)
	\label{eq: Werner}
	\end{equation}	
    where the fidelity $F\in [0.5,1)$.
	Here 
	$$
	|\phi^{\pm}\rangle=\frac{1}{\sqrt{2}}\left(|00\rangle\pm|11\rangle\right),\quad |\psi^{\pm}\rangle=\frac{1}{\sqrt{2}}\left(|01\rangle\pm|10\rangle\right)
	$$
    are the Bell's basis.
\end{exmp}
Similar to the former example, the analytic solution can be computed by the Wootters' formula. We compute the EoF of the Werner states for various $F$. The setting for the parameters are given as follows: the number of particles is $J = 500$, the noise levels are given by $\sigma = 0.01$ for Unitary preserving CBO method and $\sigma = 0.06$ for the other CBO methods.

The numerical results are shown in Fig. \ref{fig:werner}.
From the left figure, we observe that all CBO algorithms produce numerical solutions that match well with the analytical solution.
By comparing the error between the numerical solutions and the analytical solution, we conclude from the right figure that the structure preserving methods (Hermitian-preserving or unitary-preserving) are better than the corresponding projection methods in computing the EoF of Werner states.
\begin{figure}
    \centering
    \includegraphics[width=0.8\linewidth]{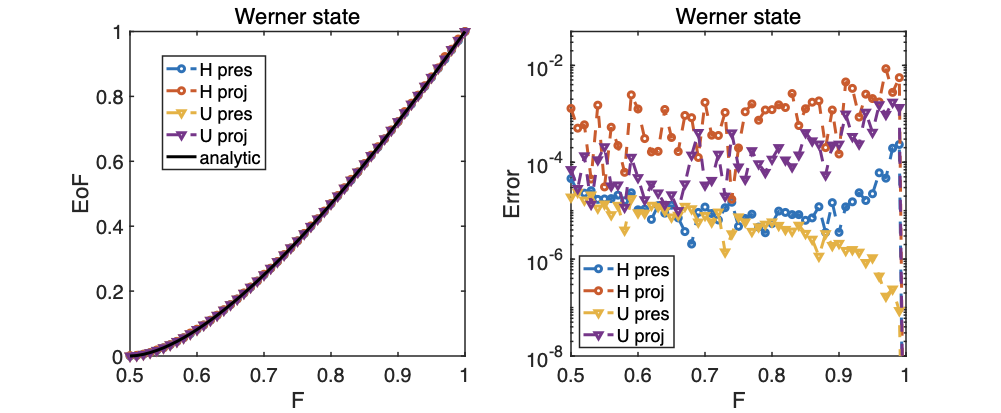}
    \caption{Comparison of different CBO methods for Werner states \eqref{eq: Werner}. Left: EoF of Werner states, Right: error between the numerical solutions and the analytical solution.}
    \label{fig:werner}
\end{figure}

\begin{exmp}[The Isotropic states $(3\times 3)$ \cite{TerhalVollbrecht00}]
	Consider the parameterized $3\times 3$ isotropic states, the density matrix is 
\begin{equation}
	\rho_{AB}=\frac{1-F}{8}(1-|\Psi^+\rangle\langle \Psi^+|)+F|\Psi^+\rangle\langle \Psi^+|,
\label{eq: isotropic}
\end{equation}
	with the maximally entangled state $$|\Psi^+\rangle=\frac{1}{\sqrt{3}}\sum_{i=1}^3|ii\rangle, \quad \left(E(|\Psi^+\rangle)=\log_23\right)$$
	where the fidelity $F\in (0,1)$.
\end{exmp}
The analytic solution is given in \cite{TerhalVollbrecht00},
$$
\mathrm{EoF}(\rho_{AB})=\left\{
\begin{aligned}
& 3(F-1)+\log_23, \quad &F>8/9,\\[2mm]
	&H(\gamma(F))+1-\gamma(F), \quad &1/3<F\leq 8/9, \\[2mm]
	&0,\qquad &F\leq 1/3,
\end{aligned}\right.
$$
where $H$ is the binary entropy and
$$
\gamma(F)=\frac{1}{3}\left(\sqrt{F}+\sqrt{2(1-F)}\right)^2.
$$

For bipartite qutrits state, $N=N_A\times N_B=3\times 3 = 9$. Also, we compute that the density matrix $\rho_{AB}$ has rank $r=9$. Hence the dimension of the problem ranges from $r\times r = 81$ to $N^2 \times N^2 = 6561$. Even in the simplified case, we have the upper bound $2N\times 2N = 324$. Therefore, the optimization problem is very high dimensional. In this example, we take the parameters in the CBO as those in the previous case. The results are exhibited in Fig. \ref{fig:isotropic}. Even for this high-dimensional problem, we find that the CBO methods produce reasonable numerical results, although the errors in EoF are larger than in the previous, simpler example. From the right figure, we still observe that the structure-preserving CBO methods achieve better performance than the projection methods.

\begin{figure}
    \centering
    \includegraphics[width=0.8\linewidth]{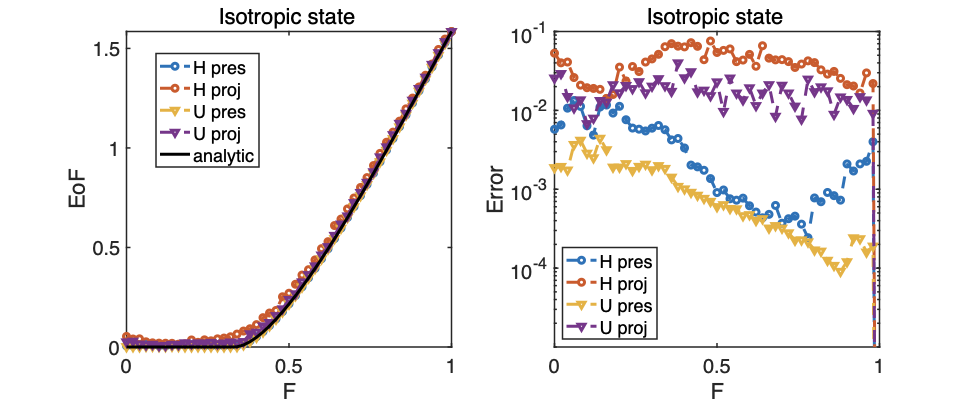}
    \caption{Comparison of different CBO methods for Isotropic states \eqref{eq: isotropic}. Left: EoF of Isotropic states, Right: error between the numerical solutions and the analytical solution.}
    \label{fig:isotropic}
\end{figure}

\begin{exmp}[The Horodecki states $(2\times 4)$ \cite{Horodecki97}]
Consider the parameterized $2\times 4$ Horodecki state, the density matrix is 
\begin{equation}
	\rho_{AB}=\frac{1}{1+7b}
\begin{pmatrix}
A & B\\[1mm]
B^T & C
\end{pmatrix}, \quad b\in[0,1],
\label{eq: Horodeckis}
\end{equation}
where $$
A=bI_4,\quad
B=b\begin{pmatrix}
	0 & 1 & 0 & 0 \\
	0 & 0 & 1 & 0 \\
	0 & 0 & 0 & 1\\
	0 & 0 & 0 & 0
\end{pmatrix},\quad 
C=\begin{pmatrix}
	\frac{1+b}{2} & ~0~ & ~0~ & \frac{\sqrt{1-b^2}}{2} \\[1mm]
	0 & ~b~ & ~0~ & 0 \\[1mm]
	0 & ~0~ & ~b~ & 0\\[1mm]
	\frac{\sqrt{1-b^2}}{2} & ~0~ & 0 & \frac{1+b}{2}
\end{pmatrix}.
$$
\end{exmp}
In this example, $b$ is a parameter. 
When $b=0$ or $b=1$, $\rho_{AB}$ is separable PPT state. For $b\in(0,1)$, $\rho_{AB}$ has rank $5$ and is an entangled PPT state. 

This is another example of a high-dimensional problem. We have $N = N_A\times N_B= 2\times 4 = 8$ and $r=5$.
Therefore, the dimension of the matrices involved in the computation ranges from $r\times r=25$ to $2N\times 2N=256$ in practice. Since no analytical solution is available for these types of states, we use the results obtained via the simulated annealing (SA) method introduced in \cite{Zyczkowski99} as a reference. For this rejection sampling method, we take the setting of parameters in \cite{Zyczkowski99}: the initial angle and the final of angle are chosen as $\chi_0 = 0.3, \chi_{end} = 10^{-4}$; the angle reduction coefficient is $\alpha = \frac{2}{3}$. In order to do a fine search, the number of iterations with angle fixed $I_{change} = 1000$ and the number of realization is $I_{mat}= 20$.

In Fig. \ref{fig: Horodeckis}, we present the numerical results computed using our CBO methods with the same parameter settings as before. In this example, we observe that the numerical solutions obtained by the projection methods deviate significantly from the reference solution, whereas those computed by the structure-preserving methods remain close to the reference solution. This is also reflected in the error plot (right figure), which illustrates the deviation of the numerical solutions from the reference solution.
\begin{figure}
    \centering
\includegraphics[width=0.9\linewidth]{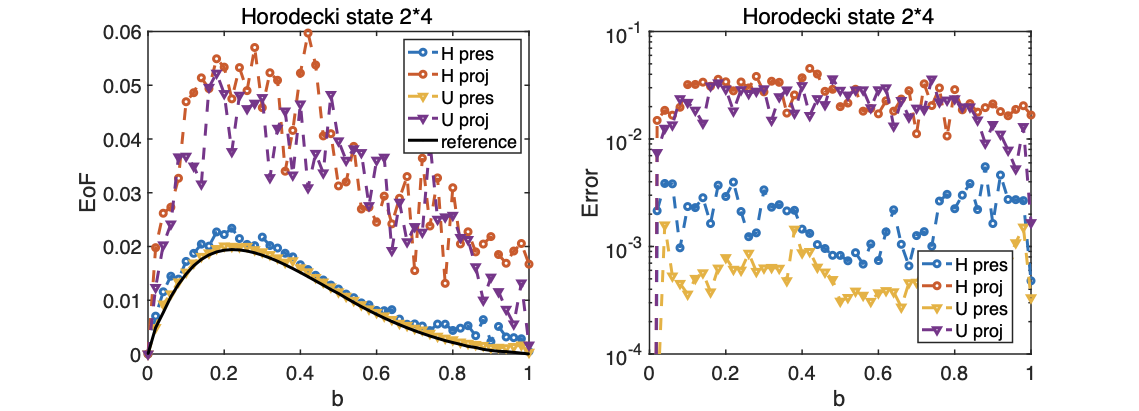}
    \caption{Comparison of different CBO methods for $2\times 4$ Horodecki states \eqref{eq: Horodeckis}. Left: EoF of $2\times 4$ Horodecki states, Right: error between the numerical solutions and the reference solution.}
    \label{fig: Horodeckis}
\end{figure}

\subsection{Numerical experiments for the Multi-species CBO method}

In this subsection, we demonstrate both the validity and the flexibility of the Multi-species CBO method by considering how particles from different species (with varying dimensions) communicate when computing the consensus.


We first consider the Werner state \eqref{eq: Werner} with fidelity $F = 0.7$. Take the computing parameters as $J = 100, K = 500, \beta = 200, \lambda = 1, \sigma = 0.06, \Delta t = 0.2$. The trajectories over time of the Hermitian-preserving CBO method and the Multi-species CBO method are shown in Fig. \ref{fig:compare_werner}. From the right figure, we observe that trajectories for different values of $M$ overlap, which implies the contribution of information exchange. 
We also observe that the consensus converges faster in the Multi-species case. Additionally, the Multi-species CBO tends to converge to the same consensus across different values of $M$, whereas the standard CBO does not (see Table \ref{tab:1}). The Multi-species method are able to attains a lower minimum (see Fig. \ref{fig:conv_werner} for the error comparison).
\begin{table}[h]\label{tab:1}
    \centering
     \caption{Minimum consensus value for different $M$ for Werner state with $F=0.7$ when $J = 100$. The real EoF for this state is $0.25022$.}
    \begin{tabular}{|c|ccccc|}\hline
        minimum & $M=4$& $M=5$& $M=6$& $M=7$ &$M=8$    \\ \hline
        H pres. CBO & 0.25026& 0.25031& 0.25052 & 0.25067 & 0.25087\\
        Multi-species H pres. CBO & 0.25023 & 0.25024 & 0.25025 & 0.25025& 0.25025\\\hline
    \end{tabular}
\end{table}

\begin{figure}
    \centering
    \includegraphics[width=0.8\linewidth]{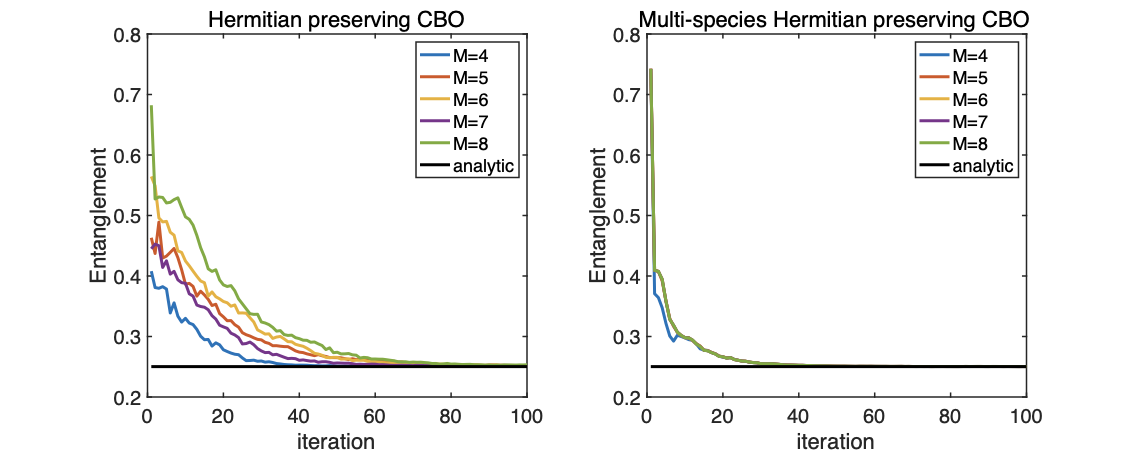}
    \caption{Trajectory of entanglement at consensus for different $M$ during iteration for Werner state with $F=0.7$ when $J = 100$.}
    \label{fig:compare_werner}
\end{figure}

\begin{figure}
    \centering
    \includegraphics[width=0.4\linewidth]{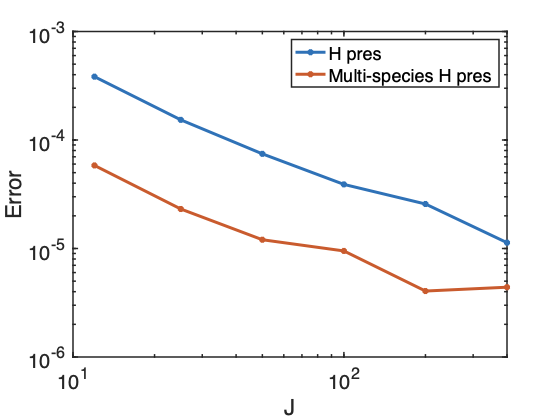}
    \caption{Errors of Hermitian preserving CBO method and the Multi-species CBO method for Werner states \eqref{eq: Werner} with $F = 0.7$. Here $J$ is the number of particles for each $M$.}
    \label{fig:conv_werner}
\end{figure}

Next, we compare the performance of the Hermitian-preserving CBO method and the Multi-species CBO method across all fidelity values $F$. The results are presented in Fig. \ref{fig:werner_multi} for Werner states \eqref{eq: Werner} and in Fig. \ref{fig:isotropic_multi} for isotropic states \eqref{eq: isotropic}.
From both figures, we observe that for most fidelity $F$, the Multi-species CBO is more accurate than the original version. So the CBO method do benefit from the communication across dimensions.
\begin{figure}[h]
    \centering
    \includegraphics[width=0.8\linewidth]{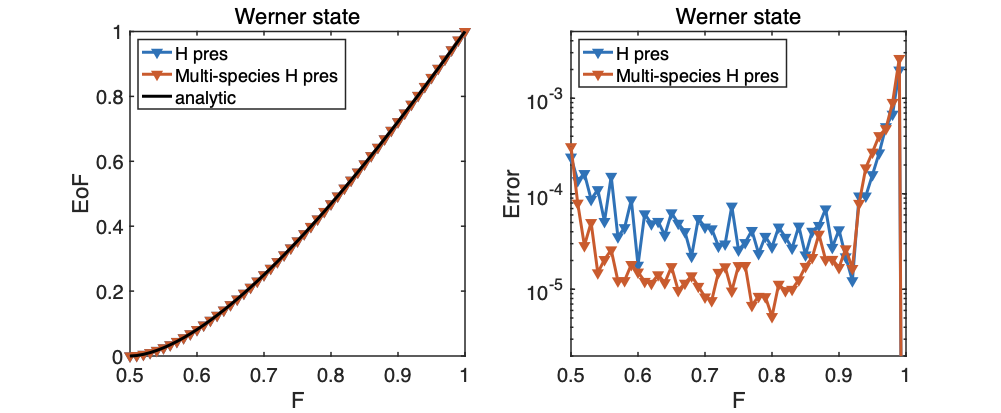}
    \caption{Comparison of Hermitian preserving CBO method with the Multi-species CBO method for Werner states \eqref{eq: Werner} when $J = 100$. Left: EoF of Werner states, Right: error between the numerical solutions and the analytical solution.}
    \label{fig:werner_multi}
\end{figure}

\begin{figure}[h]
    \centering
    \includegraphics[width=0.8\linewidth]{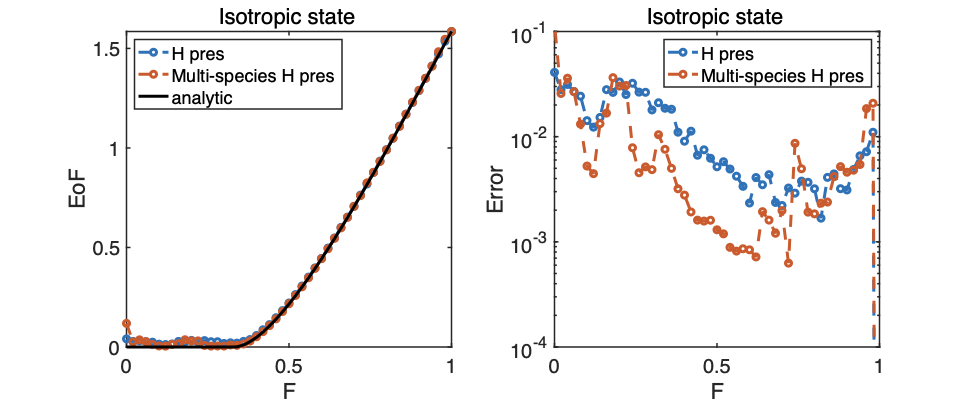}
    \caption{Comparison of Hermitian preserving CBO method with the Multi-species CBO method for Isotropic states \eqref{eq: isotropic} when $J = 100$. Left: EoF of Isotropic states, Right: error between the numerical solutions and the analytical solution.}
    \label{fig:isotropic_multi}
\end{figure}

\section*{Acknowledgments}
The authors would like to thank Prof. Zhihao Ma and Prof. Shi Jin for introducing the topic and helpful discussions.


\bibliographystyle{siamplain}
\bibliography{references}
\end{document}